\documentclass[notitlepage]{article}

\usepackage[bookmarksopen,bookmarksnumbered]{hyperref}
\usepackage[nomarkers, nolists]{endfloat}

\usepackage{natbib, amsmath, amsfonts, subcaption,amsthm, times,
  natbib, tikz, placeins, enumitem}
\usepackage{authblk}
\usepackage[margin = 1in]{geometry}

\usetikzlibrary{shapes,arrows,positioning}

\usepackage[plain,noend]{algorithm2e}

\newtheorem{proposition}{Proposition}

\newcommand{\pa}{{\rm pa}}       
\newcommand{\sib}{{\rm sib}}	   
\newcommand{\ve}{{\rm vech}}       

\begin{document}

\title{Empirical Likelihood for Linear Structural Equation Models with Dependent Errors}

\author[1]{Y. Samuel Wang}
\author[1]{Mathias Drton}

\affil[1]{Department of Statistics\\ University of Washington\\ Seattle, WA 98103 USA}

\maketitle

\begin{abstract}
  We consider linear structural equation models that are associated with
  mixed graphs.  The structural equations in these models only involve
  observed variables, but their idiosyncratic error terms are allowed
  to be correlated and non-Gaussian.  We propose empirical likelihood
  (EL) procedures for inference, and suggest several modifications,
  including a profile likelihood, in order to improve tractability and
  performance of the resulting methods.  Through simulations, we show that
  when the error distributions are non-Gaussian, the use of EL and the
  proposed modifications may increase statistical efficiency and
  improve assessment of significance.
  \\ \\
  \textbf{Keywords}: Empirical likelihood, causal inference, graphical model,
     structural equation model
\end{abstract}

\section{Introduction}

Structural equation models (SEMs) are multivariate statistical models
in which each considered variable is a function of other variables and
a stochastic error term.  Often, some of these other variables are
latent \citep{bollen2014structural}.  This paper, however, focuses on
SEMs in which the effects of latent variables are summarized.
Adopting the dominant linear paradigm, we will thus be concerned with
models in which linear functions relate only observed variables, but
error terms may be dependent.  Such models are sometimes referred to
as semi-Markovian \citep{shpitser2006identification}.  Avoiding any
explicit specification of latent confounding, the models play an
important role in exploration of cause-effect structures
\citep{colombo:2012,pearl:2009,richardson:2002,spirtes2000causation,wermuth:2011}.
Much insight about the models can be gained from a natural graphical
representation by mixed graphs/path diagrams that originates in work
of \cite{wright1921correlation}.

Formally, let $Y_1,\dots,Y_n$ be a multivariate sample with each
observation indexed by a set $V$.  So, $Y_i=(Y_{vi})_{v\in V}$ with each $Y_{vi}$ real-valued.  Now consider the system of
structural equations
\begin{equation}
\label{eq:multivarRep}
Y_{vi} = \mu_v + \sum_{u \in V \setminus v} \beta_{vu}Y_{ui} +
\epsilon_{vi}, \quad v\in V, \quad i=1,\dots,n,
\end{equation}
where the $\mu_v$ and $\beta_{vu}$ are unknown parameters and the
$\epsilon_{vi}$ are random errors.  Define vectors
$\epsilon_i=(\epsilon_{vi})_{v\in V}$ and $\mu=(\mu_v)_{v\in V}$, and
a matrix $B=(\beta_{vu})_{v,u\in V}$ with $\beta_{vu}=0$ if $v=u$.  We
assume that the error vectors $\epsilon_i$ are independent and
identically distributed, have zero means, and have covariance matrix
$\mathbb{E}\left(\epsilon_i \epsilon_i^t \right) =
\Omega=(\omega_{vu})$.  However, we do not specify any parametric form
for their distribution.  For each $i$, the equations
in~(\ref{eq:multivarRep}) can be written as
$Y_i = \mu + BY_i + \epsilon_i$.  If $(I - B)$ is non-singular, then
this system is solved uniquely by $Y_i=(I-B)^{-1}(\mu+\epsilon_i)$.
This solution has mean vector $(I-B)^{-1}\mu$ and covariance
matrix
\begin{equation} \label{eq:y_distribution}
\Sigma(B, \Omega) \;:=\; (I - B)^{-1}\Omega (I - B)^{-t}.
\end{equation}
Specific models are now obtained by hypothesizing that a particular
collection of coefficients $\beta_{vu}$ and error covariances $\omega_{vu}$
is zero.

An SEM can be represented conveniently by a path diagram/mixed graph
$\mathcal{G} = (V, E_\rightarrow, E_\leftrightarrow)$.  Here, the
vertex set $V$ yields a correspondence between the nodes of the graph
and the observed variables.  The set $E_\rightarrow$ is a set of
\emph{directed edges} $u \rightarrow v$, which encode that variable
$u$ may have a direct effect on variable $v$.  The set
$E_\leftrightarrow$ comprises \emph{bidirected edges}
$u \leftrightarrow v$ that indicate that the errors $\epsilon_{ui}$
and $\epsilon_{vi}$ may be correlated.  Define the set of
\emph{parents} of node $v$ as
$\pa(v)=\{u\in V: u\to v\in E_{\rightarrow}\}$.  Similarly, define a
set of \emph{siblings} as
$\sib(v)=\{u\in V: u\leftrightarrow v\in E_\leftrightarrow\}$.
Bidirected edges have no orientation, and $v\in\sib(u)$ if and only if
$u\in\sib(v)$.  Now, the graph $G$ induces a model through the
requirement that
\begin{align}
  \label{eq:par-space-B}
  B\in \mathcal{B}(G)&\;:=\; \left\{B\in\mathbb{R}^{V\times V}:\det(I-B) \neq 0,\;
                       \beta_{vu} = 0 \ \mathrm{if} \  u \not \in \pa(v)\right\},\\
    \label{eq:par-space-Omega}
  \Omega\in\mathcal{W}(G)&\;:=\; \left\{\Omega\in\mathbb{R}^{V\times
                           V}:\Omega \ \mathrm{pos.\ def.},\; \omega_{vu} = \omega_{uv}=0 \
                           \mathrm{if}\
                           v \not \in \sib(u)\right\}.
\end{align}
  We emphasize
that our treatment allows the model to have feedback loops, that is, $G$ may have directed cycles.

\subsection{Related Work}


Frequently, the errors in a SEM, and consequently also the
observations $Y_i$, are assumed to be multivariate Gaussian which
yield \emph{maximum likelihood estimates} (MLEs).  The Gaussian likelihood
is often maximized using generic optimization methods; as done in the
popular packages \texttt{sem} \citep{fox2017sem} and \texttt{lavaan}
\citep{lavaan} for \texttt{R} \citep{r2017}.  The coordinate-descent
methods proposed by \cite{drton2009computing} and
\cite{drton2016computation} can be a useful computational alternative
that largely avoids convergence issues.

As a less parametric method, \emph{generalized least squares} (GLS)
 minimizes a discrepancy
 between the sample covariance and the covariance implied by the
 parameters. Although the estimates are slightly more robust to
 misspecification, they are still asymptotically equivalent to the
 Gaussian MLEs \citep{olsson2000performance}.  When multivariate
 Gaussianity is inappropriate, MLEs and GLS generally lose statistical
 efficiency and yield incorrectly calibrated confidence
 intervals. \emph{Weighted least squares} methods (WLS)---also called
 \emph{asymptotically distribution free}---weight the discrepancy between the
 observed and hypothesized covariance structure by explicitly
 estimated fourth moments. Although WLS estimates are consistent and
 produce asymptotically correct confidence intervals even with
 non-Gaussian errors, the estimation of higher order moments may come at
 a loss of statistical efficiency and cause convergence issues, which
 has limited their use \citep{muthen1992comparison}.

 \citet{chaudhuri2007estimation} propose using the \emph{empirical
   likelihood} (EL) of \cite{owen2001empirical} to estimate a
 covariance matrix with structural zeros.  In our setup, this
 corresponds to the special case of a mixed graph with no directed
 edges.  \citet{kolenikov2009empirical} use EL to estimate the
 parameters of a linear SEM.  In contrast to the mixed graph
 formulation, \citet{kolenikov2009empirical} consider the case where
 the latent variable structure is explicitly modeled and all errors
 are independent.  The EL approach is appealing as it gives consistent
 estimates and asymptotically correct confidence intervals even when
 the errors are not multivariate Gaussian. However, EL can present
 numerous practical difficulties when the sample size is small
 relative to the number of parameters or estimating equations used.
 Moreover, standard implementation of EL methods is computationally
 feasible only for systems with a handful of variables.  We believe
 that these issues have prevented application of EL to linear SEMs
 beyond what was done by \citet{kolenikov2009empirical}.



\subsection{Contribution}

In this article, we apply the empirical likelihood framework to SEMs
represented by mixed graphs and propose several modifications to a
naive approach which address the most salient practical
concerns:

\begin{enumerate}[label=(\roman*)]
\item We show that in the mixed graph setting, the
covariance parameters $\Omega$ can be profiled out. This greatly
reduces the computational burden by reducing the number of estimating
equations imposed and parameters directly estimated. It also naturally
encodes the positive definite constraint on $\Omega$ and yields a
positive definite estimate of $\Omega$ for any point $B$ with a well
defined empirical likelihood.
\item When maximizing the empirical
likelihood, we leverage a recent insight and directly incorporate
gradient information in a quasi-Newton procedure instead of the
typical derivative-free approaches to empirical likelihood
optimization.  This again yields substantial computational savings.
\item We use the \emph{adjusted empirical likelihood} (AEL),
  first proposed by \citet{chen2008adjusted}.  This adjustment ensures
  that an empirical likelihood and corresponding gradient is
  well defined for every value in the parameter space.
\item We apply the idea of \emph{extended empirical likelihood}
  (EEL), which furnishes drastically improved coverage of confidence
  intervals at small sample sizes \citep{tsao2014extended}.
\end{enumerate}
Our simulations show that with these proposed modifications, empirical
likelihood becomes an attractive alternative for practitioners
concerned with non-Gaussianity in structural equation modeling.

\section{Background on Empirical Likelihood} \label{sec:empLik}

Let $Y=(Y_1,\dots,Y_n)$ be a sample from an $m$-variate distribution
$P$ belonging to a non-/semiparametric statistical model
$\mathcal{M}$.  Let $\mathcal{P}_n$ be the $n-1$ dimensional
probability simplex.  For $p=(p_1,\dots,p_n)\in\mathcal{P}_n$, define
the log-empirical likelihood $\ell(p; Y) = \sum_{i=1}^n \log(p_i)$.
This is the log-likelihood of the sample under the discrete
distribution with mass $p_i$ at each point $Y_i$.  Suppose we are
interested in a parameter $\theta=\theta(P)$ taking values in
$\Theta \subseteq \mathbb{R}^d$ such that for a map
$G: \mathbb{R}^m \times \mathbb{R}^d \mapsto \mathbb{R}^q$ we have
$\mathbb{E}_{P} G(Y_i, \theta(P)) = 0$ for all $P\in\mathcal{M}$.  The
log-empirical likelihood at a given parameter value $\theta$ is then
\begin{equation}
  \label{eq:EL:general}
\ell(\theta; Y) = \max_{p \in \mathcal{P}_\theta } \; \ell(p; Y) =
\max_{p \in \mathcal{P}_\theta } \;\sum_{i=1}^n \log(p_i),
\end{equation}
where the feasible set
\begin{equation}
  \label{eq:general:Ptheta}
\mathcal{P}_\theta = \left\{p\in\mathcal{P}_n: \sum_{i=1}^n  p_i G(Y_i, \theta) = 0 \right\}
\end{equation}
 reflects that the
expectation of $G(\cdot; \theta)$ vanishes for distributions
compatible with $\theta$.

The empirical likelihood (EL) from~(\ref{eq:EL:general}) provides a
basis for statistical inference.  Maximizing it over $\theta\in\Theta$
yields the \emph{maximum empirical likelihood estimator}
\begin{equation}
\check \theta = \arg\max_{\theta} \ell(\theta; Y)
\end{equation}
that we refer to as MELE.  Ratios of the EL yield
\emph{empirical likelihood ratio} statistics.
\citet{owen1988empirical} derives an EL analogue of Wilk's Theorem,
and the result was expanded to the general estimating equation
framework by \citet{qin1994empirical}. The specific regularity
conditions needed are discussed in Section \ref{sec:asymp}, and the
results imply under very general conditions that the MELE is
consistent and asymptotically normal. In addition, EL ratio statistics
have limiting $\chi^2$ distributions that can be used to calibrate
statistical tests and create confidence intervals or regions.  For a
detailed exposition of these ideas, we refer readers to
\cite{owen2001empirical}.

The nice theoretical properties for EL, however, come at a high
practical cost.  The practical issues become particularly pressing for
applications to linear SEMs, for which the number of parameters and
estimating equations generally grow on the order of $m^2$, where
$m=|V|$ is the number of variables considered.  We describe
three difficulties that complicate the direct use of EL for
SEMs:
\begin{enumerate}[label=(\roman*)]
\item For some values $\theta$, the origin may be outside the convex
  hull of $\{G(Y_i, \theta):i=1,\dots,n\}$, in which case the feasible
  set $\mathcal{P}_\theta$ from~(\ref{eq:general:Ptheta}) is empty and
  the EL at $\theta$ is zero.  This ``convex hull problem''
  occurs more often when the sample size is small relative to the
  number of estimating equations or when the data is skewed. As
  discussed by \citet{grendar2009empty}, it is possible that
  $\mathcal{P}_\theta=\emptyset$ for all parameter vectors $\theta$,
  which is known as the ``empty set problem''.
  In addition, the log-EL is typically not a convex function
  \citep{chaudhuri2016constrained}, and finding an initial
  point that has well-defined EL and is in the basin of attraction of
  the MELE can be difficult.

\item The optimization problem defining the log-EL $\ell(\theta;Y)$
  from~(\ref{eq:EL:general}) is typically solved iteratively through
  its dual.  Although this problem is convex, it can be
  computationally burdensome when the number of estimating equations,
  which corresponds to the number of dual variables, is large.

\item Confidence intervals based on the asymptotic normal variance and
  $\chi^2$ likelihood ratio calibration have been shown to often
  undercover at small sample sizes \citep{tsao2014extended}.
\end{enumerate}

\section{Empirical Likelihood for
  SEMs} \label{sec:elwithSecondMoments}

We now turn to the application of EL to SEMs.  For expository
simplicity, we assume throughout that our observations are centered.
In other words, the intercept parameter vector $\mu$
for~(\ref{eq:multivarRep}) is zero, so that $\mathbb{E}(Y_i)=0$.
However, our ideas extend straightforwardly to the case where we also
make inference about $\mu \neq 0$.

\subsection{Profiled Formulation}

Consider the linear SEM given by a mixed graph
$\mathcal{G}=(V,E_\rightarrow,E_\leftrightarrow)$, as defined in the
Introduction.  The general framework laid out in
Section~\ref{sec:empLik} can be applied directly to such a model by
taking the covariance matrix of the observations $Y_i$ as the general
parameter $\theta$.  We may then define an EL at a pair of parameter
matrices $(B,\Omega)$ as the EL at the covariance matrix
$\Sigma(B,\Omega)$ from~(\ref{eq:y_distribution}).  In such a direct
application to the linear SEM, the log-EL function
$\ell( B, \Omega; Y)$ is the maximum of the log-EL $\ell(p; Y)$ over
the set
\begin{equation}\label{eq:naive}
\begin{aligned}
&\mathcal{P}_{\Sigma(B,\Omega)} \;=\;\left\{p\in\mathcal{P}_n:  \sum_{i=1}^n
  p_i Y_i  = 0, \; \sum_{i=1}^n p_i\left[\ve\left(Y_iY_i^T\right) - \ve\, \Sigma(B, \Omega)\right] = 0\right\}. \\
\end{aligned}
\end{equation}
Here, $\ve$ is the half-vectorization operator for symmetric
matrices. Under this formulation, there are $m$ constraints for the
mean and $m(m + 1)/2$ covariance constraints, and the MELE is computed
by optimization with respect to the pair of $m\times m$ matrices
$(B,\Omega)$, with $\Omega$ restricted to be positive definite.

Inspection of the covariance constraints reveals that a great
simplification is possible by profiling out $\Omega$.  Indeed, the
covariance constraint yields an explicit solution for $\Omega$ given
$B$, $Y = (Y_1, \ldots Y_n)$, and $p$.  Specifically, with
$\Pi=\text{diag}(p_1,\dots,p_n)$, we have
\begin{equation}
  \label{eq:equiv-equations}
Y\Pi Y^T = \Sigma(B,\Omega) = (I - B)^{-1} \Omega(I - B)^{-t} \;\iff\; (I - B)Y \Pi Y(I-B)^{t} = \Omega.
\end{equation}
The entries of $\Omega$ are either constrained to be zero or freely
varying.  No constraints arise from the freely varying entries, and we
may base estimation of $B$ on only the structural zeros in $\Omega$,
that is,
\[\left\{(I - B)Y \Pi Y(I-B)^{t}\right\}_{uv} = 0 \quad \forall \{u,v\} \not \in E_\leftrightarrow.\]
Once a solution for $B$ is found, we may simply compute
$\Omega=\Omega(B)$ by setting
$
\omega_{uv} = \left\{(I - B)Y \Pi Y(I-B)^{t}\right\}_{uv}
$
for $u=v$ or $\{u,v\}\in E_\leftrightarrow$.  The profile log-EL in this approach
is the function
\begin{equation}
\ell(B;Y) \;=\; \max_{p\in\mathcal{P}_B}
\ell(p;Y)\label{eq:profile:inner}
\end{equation}
obtained from the set of weight vectors
\begin{equation}
\mathcal{P}_B \;=\; \left\{p\in\mathcal{P}_n:  \sum_{i=1}^n
  p_i Y_i  = 0, \; \sum_{i=1}^n p_i \left(Y_{vi} - \sum_{s \in
    \pa(v)}\beta_{vs}Y_{si}\right) \left(Y_{ui} - \sum_{t \in
    \pa(u)}\beta_{ut}Y_{ti}\right) = 0 \quad \forall\{v,u\} \not \in
  E_\leftrightarrow \right\}.\label{eq:profiled}
\end{equation}
The MELE $\check B$ is found by maximizing $\ell(B;Y)$ over the set
$\mathcal{B}(G)$ from~(\ref{eq:par-space-B}), and then
$\check\Omega=\Omega(\check B)$.  We emphasize that there are now only
$m(m - 1)/2 - |E_\leftrightarrow|$ covariance constraints, and only
the matrix $B$ needs to be optimized.

Following a standard strategy,
we evaluate $\ell(B; Y)$, that is, solve the ``inner maximization''
in~(\ref{eq:profile:inner}) at a fixed $B$, through the dual problem.
Strong duality holds because the constraints in (\ref{eq:profiled})
are linear in the weights $p_i$.  Let $G(Y_i,B)$ be the map with
coordinates $G_v(Y_i,B)=Y_{iv}$ for $v\in V$ and
$G_{uv}(Y_i,B)=g_u(Y_i,B)g_v(Y_i,B)$ for each nonedge
$\{u,v\} \notin E_\leftrightarrow$, where
$g_v(Y_i,B)=Y_{vi} - \sum_{s \in \pa(v)}\beta_{vs}Y_{si}$.  With dual
variables $\alpha\in\mathbb{R}$ and
$\lambda\in\mathbb{R}^{m+m(m - 1)/2 - |E_\leftrightarrow|}$, the
Lagrangian for the inner optimization over $\mathcal{P}_B$ is
\begin{equation}
\begin{aligned}
L_B( p, \alpha, \lambda ) & = -\sum_{i=1}^n \log(p_i) + \alpha \left(\sum_{i=1}^n
  p_i - 1\right) + n \sum_{i=1}^n p_i\left(\sum_{v\in V}
  \lambda_{v}   G_{v}(Y_i, B) + \sum_{\{u,v\}\notin E_\leftrightarrow} \lambda_{uv} G_{uv}(Y_i, B)\right).
\end{aligned}
\end{equation}
Maximizing over the weights, with $\alpha=n$, we find
\begin{equation}
  \label{eq:optimal-pi}
  \check p_i = \frac{1}{n} \frac{1}{1 + \sum_{v\in V}
  \lambda_{v}   G_{v}(Y_i, B) +\sum_{\{u,v\}\notin E_\leftrightarrow} \lambda_{uv} G_{uv}(Y_i, B)},
\end{equation}
and substitution into $L_B$ yields a convex dual
function of $\lambda$.  We optimize it via
Newton-Raphson with a backtracking line search to ensure
$0 \leq p_i \leq 1$.

In the ``outer maximization'', we optimize $\ell(B; Y)$ with respect
to $B$ using a gradient based quasi-Newton method.  Although we can
only evaluate $\ell(B; Y)$ numerically, once we have the optimal dual
variables $\lambda$ and the corresponding weights
from~(\ref{eq:optimal-pi}), we can analytically compute the gradient
of $\ell(B; Y)$ as
\begin{equation}
\nabla \ell(B; Y) = - \lambda^T\sum_{i=1}^n \check p_i \nabla G(Y_i,
B);
\end{equation}
see \citet{chaudhuri2016constrained}.
The
Hessian, however, cannot be computed in closed form, so we use BFGS
which builds an approximate Hessian via the gradient.

Although both formulations yield the same MELE, the profile approach
from (\ref{eq:profile:inner}) and (\ref{eq:profiled}) drastically
eases difficulties (i) and (ii) discussed in
Section~\ref{sec:empLik} as the number of estimating equations for the
covariance is reduced to $m + m(m - 1)/2 - |E_\leftrightarrow|$.  This
reduces the number of dual variables to optimize in the inner
maximization.  Moreover, when profiled, the outer maximization
searches over only $B\in\mathcal{B}(G)$ while the naive direct
formulation from~(\ref{eq:naive}) requires a search over both
$B\in\mathcal{B}(G)$ and $\Omega\in\mathcal{W}(G)$; in particular,
positive definiteness of $\Omega$ needs to be respected in the naive
optimization.  Finally, satisfying the convex hull condition for the
error covariances typically requires a simultaneous good choice of $B$
and $\Omega$. The directed edge weights can be easily initialized with
regression estimates, but the covariance parameters are typically more
difficult to specify.  In Section \ref{sec:simulations}, we show that
the computational advantages produce substantial gains in computation
time and converge to a valid stationary point at a much higher
proportion of the time even when the sample size is small.

\subsection{Small Sample Improvements}
\label{sec:small-sample-impr}

In addition to reformulating the optimization problem, we make two
modifications to improve the performance of EL for SEMs.  We apply
\emph{adjusted empirical likelihood} (AEL) to improve the search for a MELE
and use \emph{extended empirical likelihood} (EEL) to improve the coverage of
confidence intervals.

\citet{chen2008adjusted} proposed AEL to alleviate the convex hull
problem mentioned in difficulty (i) above.  The adjustment amounts
to adding a pseudo-observation whose contribution to the estimating
equations is
$G_{n+1}(B) = -a_n \bar G(B) = -a_n \frac{1}{n}\sum_{i=1}^n G(Y_i, B)$
for a choice of $a_n>0$.  Adding this term ensures that no matter the
value of $B$, the set of feasible weight vectors, now in $\mathcal{P}_{n+1}$, is
non-empty.  Hence, the log-AEL $\ell^a(B; Y)$ and its gradient
\begin{equation}
  \nabla \ell^a(B; Y) = - \lambda^T\sum_{i=1}^n \left[\check p_i +\left(-\frac{a_n}{n}\right)\check p_{n+1} \right] \nabla G(Y_i, B)
\end{equation}
are
well defined across the entire parameter space.
\citet{chen2008adjusted} show that AEL retains the asymptotic
properties of the original EL when $a_n = o(n^{2/3})$, and suggest
$a_n = \log(n) / 2$.  We adopt this choice.

The terms in our covariance constraint are products,
$G_{uv}(Y_i,B)=g_v(Y_i,B)g_u(Y_i,B)$.  This is generally not true for
the added term $G_{n+1}(B)$ and is not clear how to define an
appropriately sparse and positive definite matrix
$\Omega(B)$ using AEL weights.  Thus, we propose finding an
estimate $\check B$ that maximizes the AEL and computing
$\check \Omega=\Omega(\check B)$ based on weights from
recalculating the original EL at $\check B$.  As demonstrated in
our numerical experiments, this approach alleviates some convergence
issues but, of course, the original EL may be zero at the AEL
maximizer $\check B$, in which case we do not have an estimate of
$\Omega$ and say that the AEL procedure has not converged.

To address undercoverage of confidence regions for smaller samples, as
described in difficulty (iii), we adopt the EEL of
\citet{tsao2014extended} who show that their $\chi^2$-calibrated EEL
confidence regions outperform those from the original EL.  Assuming
the MELE $\check B$ exists, a positive EEL may be defined for any
matrix $B\in\mathcal{B}(G)$ by taking the original EL at a convex
combination of $B$ and $\check B$.  Specifically, the log-EEL
suggested by \citet{tsao2014extended} is
\begin{equation}
\ell^e(B; Y) = \ell(h^{-1}(B, Y), Y)
\end{equation}
for $ h(B, Y) = \check B + \gamma(n, \ell(\theta; Y))(B - \check B)$
with $\gamma(n, \ell(B; Y)) = \left(1 + \frac{2(-n\log(n) - \ell(B; Y))}{2n}\right)$.


\subsection{Asymptotic Distribution of Empirical Likelihood
  Estimators} \label{sec:asymp}

It follows from \citet[Theorem 1]{qin1994empirical} that under the
following assumptions, MELEs are asymptotically normal and empirical
likelihood ratios converge to $\chi^2$ limits.  The same is true for
the modifications from Section~\ref{sec:small-sample-impr}.

\begin{proposition}\label{lem:conditions}
  Let $G=(V,E_\rightarrow,E_\leftrightarrow)$ be a mixed graph, let
  $B_0\in\mathcal{B}(G)$ and $\Omega_0\in\mathcal{W}(G)$.  Let
  $\epsilon$ be a zero-mean random vector with covariance matrix
  $\Omega_0$.  Assume that:
  \begin{enumerate}[label=(\alph*)]
  \item The Jacobian of the parametrization
    $(B,\Omega)\mapsto \Sigma(B,\Omega)$ defined on
    $\mathcal{B}(G)\times\mathcal{W}(G)$ has full rank at
    $(B_0,\Omega_0)$.
  \item The joint distribution of $\epsilon$ and
    $\epsilon^{(2)}=(\epsilon_v\epsilon_u: v,u\in V)$ is
    non-degenerate and has finite third moments.
  \end{enumerate}
  If $Y_1,\dots,Y_n$ is an i.i.d.~sample from the distribution
  determined by $(B_0,\Omega_0,\epsilon)$, i.e., the distribution of
  $(I-B_0)^{-1}\epsilon$, then the MELE
  $\check \theta = \left(\ve \left[ \check B \right], \ve \left[
      \Omega (\check B) \right] \right)$ is asymptotically
  normal with
\begin{equation}
\begin{aligned}
  \sqrt{n}\left(\check \theta - \theta_0\right) \rightarrow N(0, V), \qquad
  V^{-1} = \mathbb{E}\left( \frac{\partial G(Y, \theta_0)}{\partial \theta}\right)^t \mathbb{E}[G(Y, \theta_0)G(Y, \theta_0)^t]^{-1} \mathbb{E}\left( \frac{\partial G(Y, \theta_0)}{\partial \theta}\right).
\end{aligned}
\end{equation}
Here, $G$ is given by the estimating equations corresponding to the
naive formulation in (\ref{eq:naive}). Furthermore, EL
ratio statistics have $\chi^2$ limits.  In particular, for
$q = m + m(m + 1)/2$ and
$d = |E_\rightarrow| + |E_\leftrightarrow| + m$, we have
\begin{equation}
\begin{aligned}
2 \left(-n\log(n) - \ell(\check \theta; Y)  \right) &\rightarrow \chi^2_{(q - d)}, & \qquad \qquad  & 2 \left[\ell(\check \theta; Y) -  \ell(\theta_0; Y) \right] &\rightarrow \chi^2_d.
 \end{aligned}
 \end{equation}
\end{proposition}

We sketch the proof of the proposition in the appendix.

If the rank condition from (a) holds, then the rational map
$(B,\Omega)\mapsto \Sigma(B,\Omega)$ has full rank Jacobian at almost
all choices of $(B,\Omega)$, and the map is generically finite-to-one.
There is thus a connection to local/finite identifiability of
$(B,\Omega)$ from the covariance matrix.  For state-of-the-art methods
for determining identifiability see
\cite{foygel2012half,chen2016old,drton:weihs:2016}.

\section{Numerical Simulations} \label{sec:simulations}
We now show a series of numerical experiments to evaluate the effectiveness of the proposed methods and compare the results to existing methods.

\subsection{Convergence of Optimizers for Naive vs Profile Formulation}\label{sec:naiveProf}
We first compare the naive/direct procedure which explicitly estimates
$B$ and $\Omega$ to the profiled procedure which only involves
$B$. For both procedures, we use the original EL and adjusted EL.  We
also consider a hybrid method, which first finds the maximum AEL point
to initialize a search which then uses original EL. We randomly
generate acyclic mixed graphs with 8 nodes, 10 directed edges, and 6
bidirected edges. We randomly select directed edges $u \rightarrow v$
from all pairs such that $u < v$ and then select bidirected edges
$u \leftrightarrow v$ from the remaining unselected pairs.  This setup
ensures that $(B,\Omega)$ are generically identifiable from
$\Sigma(B,\Omega)$ by the result of \cite{brito:2002}.

We generate random true parameter matrices $B=(\beta_{uv})$ and
$\Omega=(\omega_{uv})$ as follows.  The coefficients $\beta_{uv}$ are
drawn uniformly from $(-1, .-.2) \cup (.2, 1)$.  For $\Omega$, we draw
off-diagonal elements $\omega_{uv} = \omega_{vu}$, $u\not=v$,
uniformly from $(-.8, .-.3) \cup (.3, .8)$.  We then use exponential
draws to set
$\omega_{vv} = \sum_{u \neq v} |\omega_{uv}| + 1 + \text{exp}(1)$.

We consider errors from four distributions.  First, we generate
centered multivariate Gaussian errors with covariance matrix $\Omega$.
Second, we generate them from a multivariate $T$-distribution with 4
degrees of freedom, which we denote by $T_4$, again with expectation
zero and covariance matrix $\Omega$.  Third, we consider log-normal
errors.  In this case, we simulate a multivariate Gaussian vector $Z$,
centered and with covariance matrix equal to the correlation matrix
$C$ that corresponds to $\Omega$. We then set the error vector to
$\epsilon = \exp(Z) - \sqrt{e}$, which yields covariance matrix
$e(\exp(C) - 1)$.  Finally, in order to draw a multivariate
distribution with recentered gamma marginals and covariance $\Omega$,
we follow the steps:
\begin{enumerate}
\item Draw $\epsilon_v \sim \text{gamma}(\text{shape} = \omega_{vv} - \sum_{v \neq u} |\omega_{uv}|, \text{scale} = 1)$.
\item For each $\{u,v\}\in E_\leftrightarrow$, generate $\delta_{uvi}
  \sim \text{gamma}(\text{shape} = |\omega_{uv}|, \text{scale} = 1)$
  and a random sign $\xi_{uv}\in\{-1, 1\}$.
\item If $\omega_{uv} > 0$, add $\xi_{uv} \delta_{uvi}$ to
  $\epsilon_{ui}$ and $\epsilon_{vi}$.  If $\omega_{uv} < 0$, add
  $\xi_{uv} \delta_{uvi}$ to $\epsilon_{ui}$ and
  $-\xi_{uv} \delta_{uvi}$ to $\epsilon_{vi}$.
\item Subtract the true mean from each error term so that it has mean 0.
\end{enumerate}
All optimizations are initialized with a procedure from
\citet{drton2016computation}, where the free elements of $B$ are
calculated via least squares.  The resulting residuals are used to
initialize the non-zero values $\omega_{uv}$.  If a row is not
diagonally dominant, the off-diagonal elements are scaled so that
$\sum_{j\neq i} |\omega_{ij}| <.9 \times \omega_{ii}$ to ensure $\Omega$
is positive definite.

\begin{figure}
  \centering
  \includegraphics{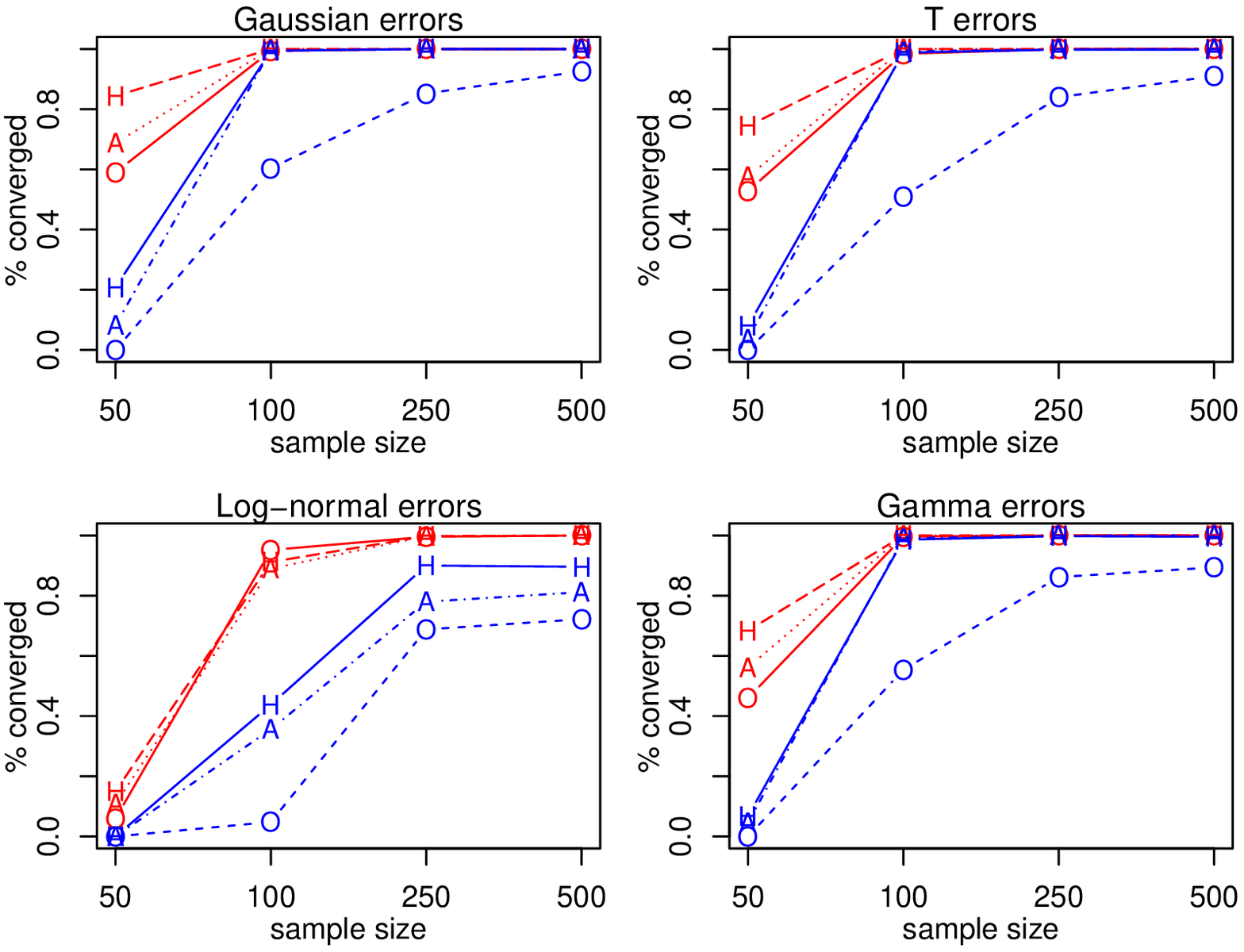}
  \caption{\label{fig:naiveProfConv}Proportion of 500 simulations
    which converge to a valid stationary point, plotted versus the
    sample size. O- original EL; A- adjusted EL; H- hybrid EL. Red
    points indicate the profile formulation; blue points indicate the
    naive formulation.}
\end{figure}

  \begin{figure}
    \centering
  \includegraphics{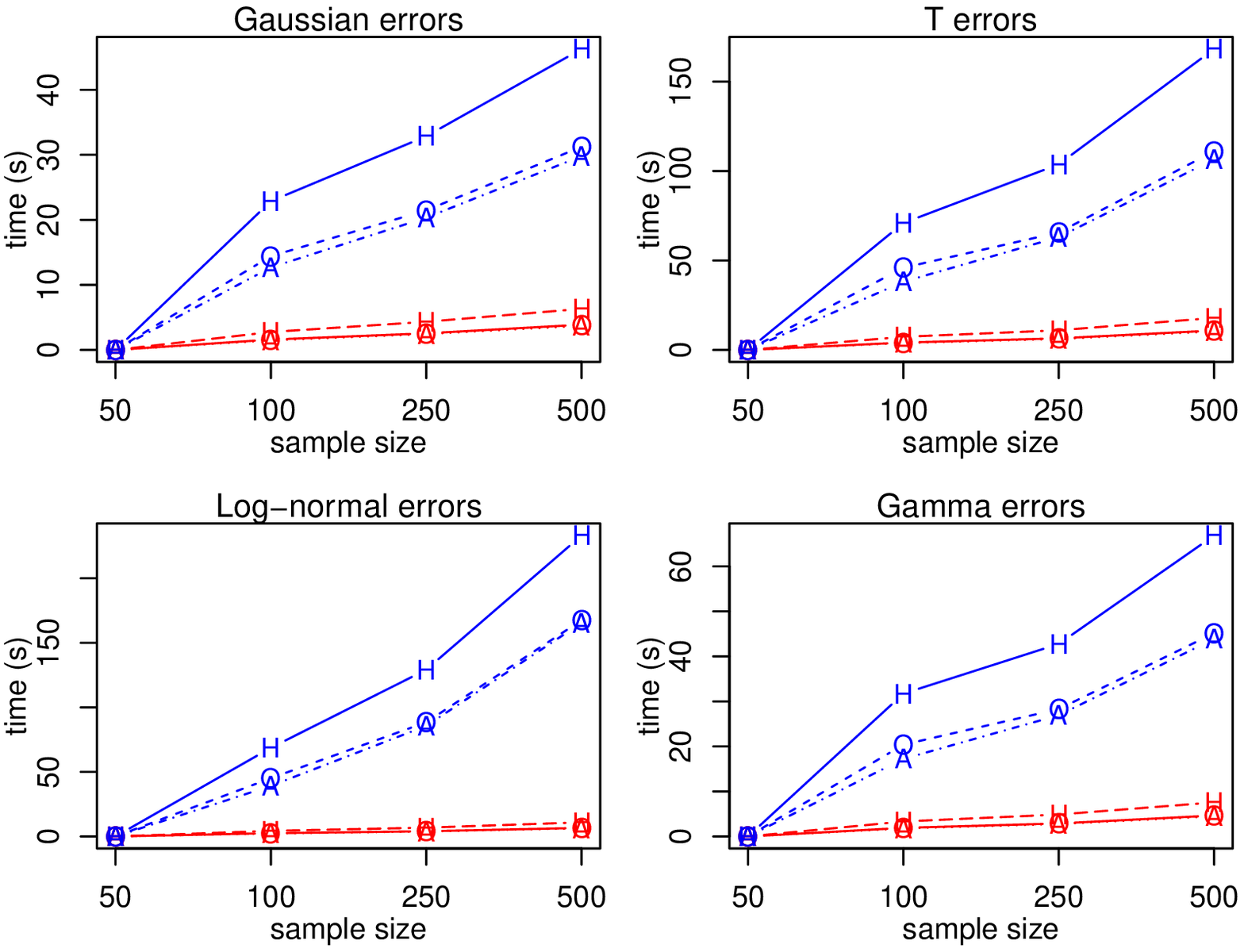}
  \caption{\label{fig:naiveProfTime}The average run time in seconds among the simulations in which all methods converge to a valid stationary point, plotted versus the
    sample size. O- original EL; A- adjusted EL; H- hybrid EL. Red points indicate the profile formulation; blue points indicate the naive formulation.}
\end{figure}

Figure \ref{fig:naiveProfConv} shows that in all cases the profiled
formulation converges at least as often as the naive formulation.  AEL
converges more often than original EL, and the hybrid procedure converges
the most often.  Even at a sample size of $n =100$, the profiled
problem converges nearly every single time, except in the case of
log-normal errors. Figure \ref{fig:naiveProfTime} shows that the profiled
form can be up to 40 times faster on average than the naive form.

\subsection{Estimation Error}

We now explore the estimation errors resulting from different
approaches.  We compare both original EL and AEL to the Gaussian MLE
computed as in \cite{drton2009computing}, GLS, and WLS.  The latter
two estimates are computed using the R package \texttt{lavaan}
\citep{lavaan}.  We also include a hybrid procedure that finds the
Gaussian MLE $\hat B$ and then uses the resulting residuals and the
maximum EL weights at $\hat B$ to form an estimate
$\check \Omega = (I - \hat B)^t Y \Pi Y ^t(I - \hat B)^t$.  Note that
the $T_4$ distribution does not have finite 6th moments, so the
limiting distributions from Proposition~\ref{lem:conditions} may not
hold; however, all estimation procedures still appear to be
consistent.

Proceeding as in Section \ref{sec:naiveProf}, we generate 1000 graphs
for each error distribution and sample size.  To measure estimation
accuracy, we average the relative error
$\Vert \ve(\check \Sigma) - \ve(\Sigma) \Vert^2/\Vert
  \ve(\Sigma)\Vert^2$ for $\Sigma(B, \Omega)$ across each of the
simulation runs in which all methods converge; recall Figure
\ref{fig:naiveProfConv}.  The results are shown in Figure
\ref{fig:msqe}.

\begin{figure}[htbp]
\centering
\includegraphics{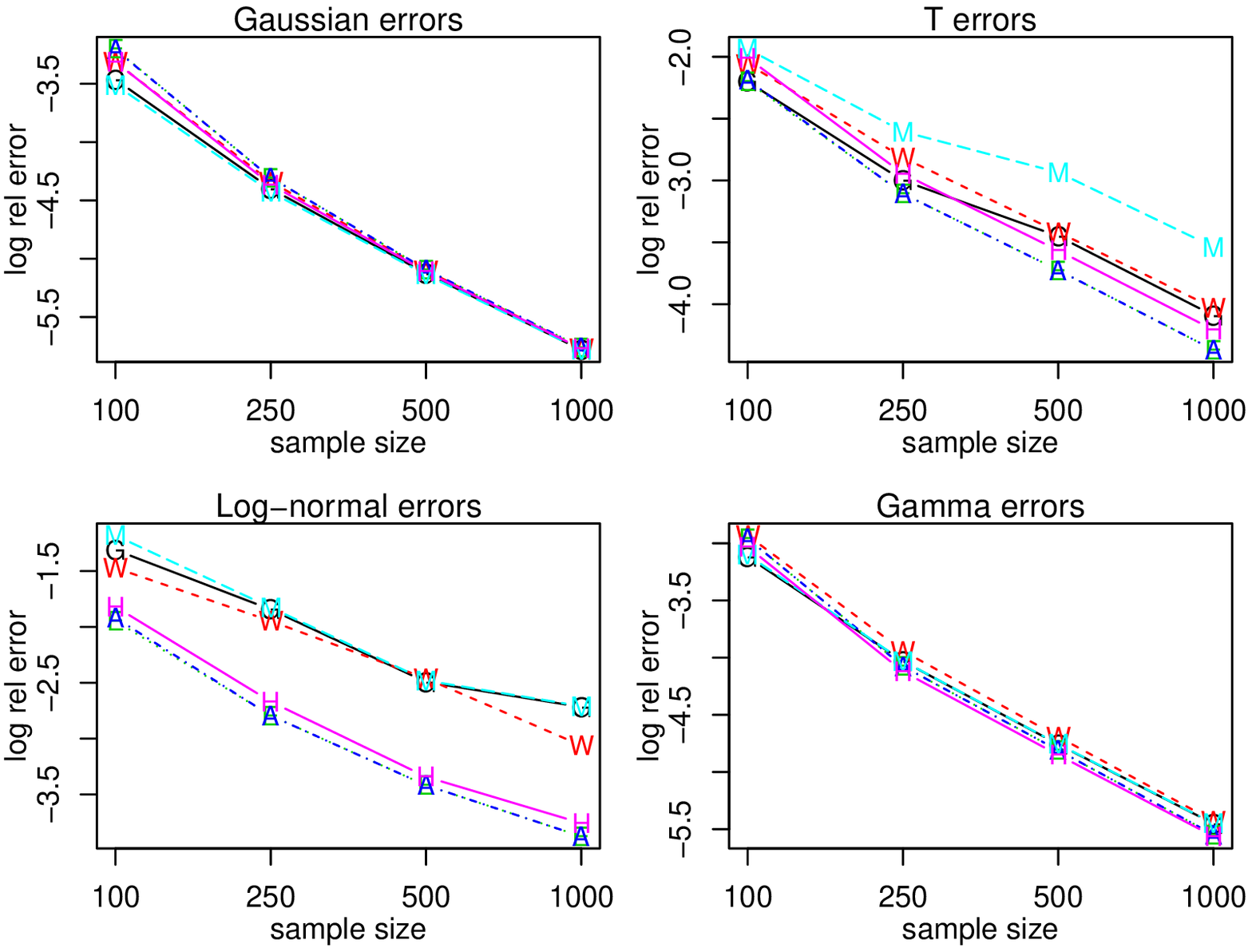}
\caption{\label{fig:msqe}Log mean relative squared estimation error in $\Sigma$ over 1000 simulations, plotted versus the
    sample size. Average is only taken on simulations in which all methods converged. A- adjusted EL; E- empirical likelihood; G- generalized least squares; H- hybrid Gauss/EL; M- Gaussian MLE; W- weighted least squares.}
\end{figure}

\begin{figure}[htbp]
\centering
\includegraphics{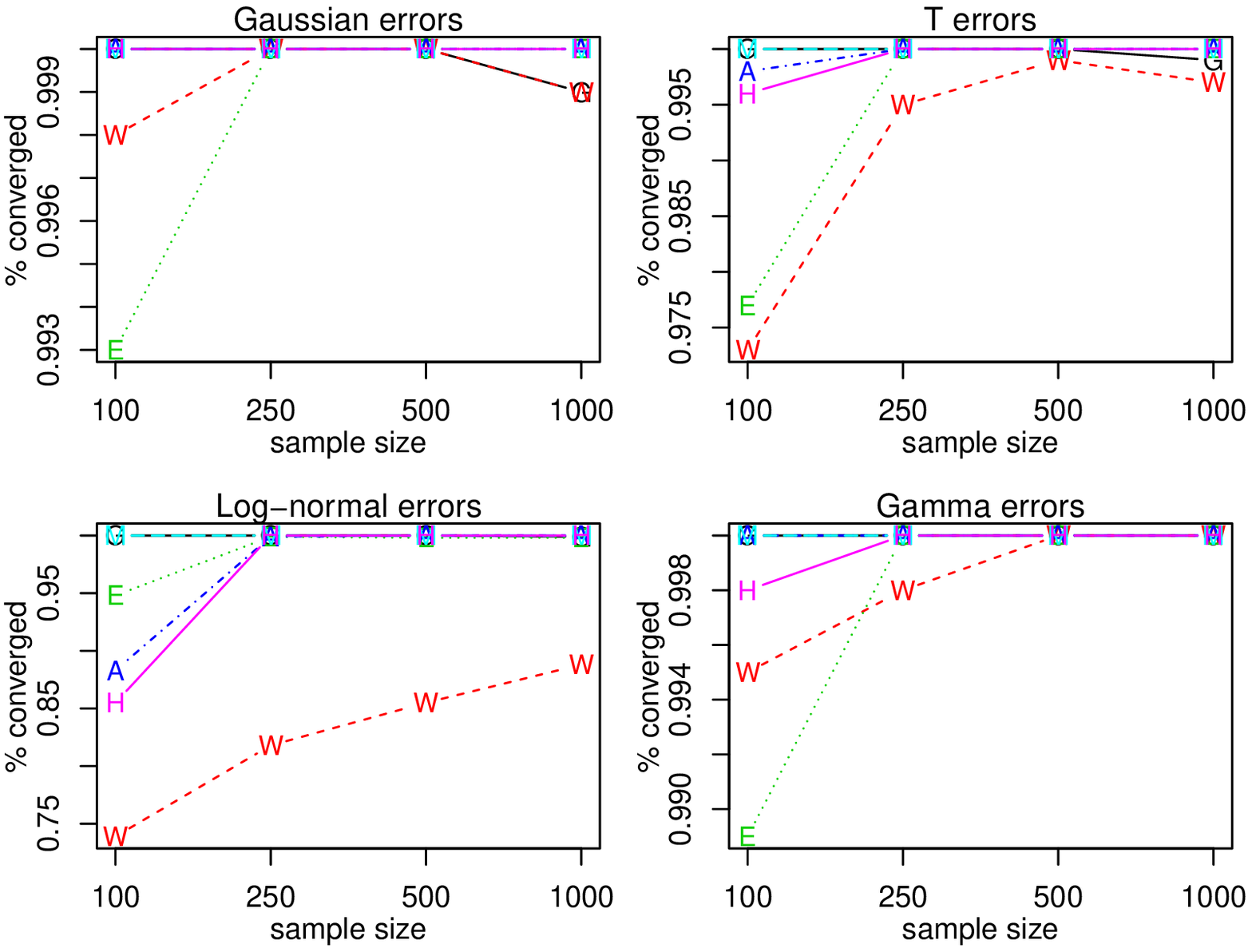}
\caption{\label{fig:compareConv}Proportion of times (over 1000 simulations) the method converged to a local maximum of the objective function, plotted versus the
    sample size. A- adjusted EL; E- empirical likelihood; G- generalized least squares; H- hybrid Gauss/EL; M- Gaussian MLE; W- weighted least squares.}
\end{figure}

In general, there is no substantial difference in accuracy between the
adjusted and original empirical likelihood methods. For the Gaussian
case, MLE and GLS perform better than the methods which do not assume
Gaussianity, but the improvement is slight. In the $T_4$ and
log-normal case, the EL procedures perform substantially better than
the other methods. Finally, for the gamma case, the hybrid method
seems to outperform the other methods, followed closely by the EL
methods; however, the differences between the methods are not
substantial.
In Figure \ref{fig:compareConv}, all methods converge
more than 95\% of the time in all distributions, except for the log-normal
case. In this case, the WLS procedure still only converges roughly
90\% at $n = 1000$.

\subsection{Confidence Regions}

We examine the coverage frequencies of joint confidence regions for
the parameters $\beta_{uv}$ and $\omega_{uv}$.  We construct Wald
regions using the estimates of $\text{Var}(\hat \theta)$ from the
Gaussian MLE, GLS, and WLS. We also calculate a sandwich variance
estimator using the Gaussian likelihood as the estimating equations
and the asymptotic EL variance via \cite{qin1994empirical}.
Alternatively, we calculate the EL at $(B_0, \Omega_0)$ using original
EL, EEL, AEL. We then compare the resulting EL
ratio to its asymptotic $\chi^2$ distribution. If a method
does not converge, we count this as a case in which the confidence
region does not cover the true parameters.

At each sample size and error distribution, we construct 1000 graphs
with 6 nodes, 8 directed edges and 4 bidirected edges from the
procedure described in Section~\ref{sec:naiveProf}. For the $T$
distribution, we increase the degrees of freedom to 7 to ensure
Proposition~\ref{lem:conditions} applies. The coverage rates for 90\%
confidence intervals are shown in Figure \ref{fig:confInt}.  Based on
the displayed results, regions obtained from the Gaussian MLE and GLS
can only be recommended when the errors are (close to) Gaussian.  The
EEL method performs the best, staying close to the parametric methods
in the Gaussian case and doing the best in most non-Gaussian
scenarios.  The sandwich method is another good choice.  However, we
also observe that in order to achieve nominal coverage levels very
large sample sizes may be required.


\begin{figure}[htbp]
\centering
\includegraphics{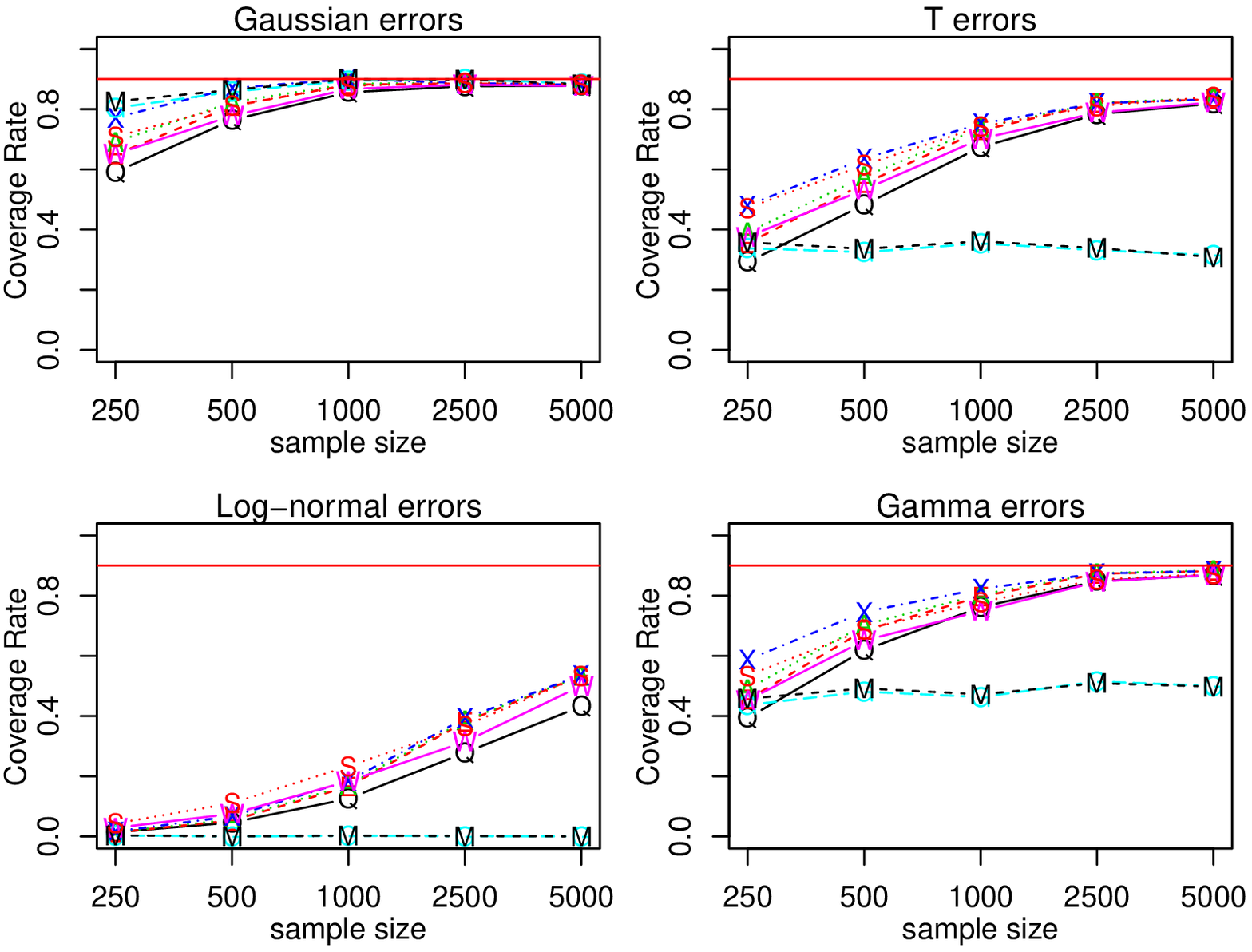}
\caption{\label{fig:confInt}Coverage frequencies of joint confidence intervals. A- adjusted EL; X- extended EL; G- generalized least squares; Q- Qin and Lawless asymptotic EL variance; M- Gaussian MLE; S- sandwich estimator; W- weighted least squares; E- EL direct $\chi^2$ calibration.}
\end{figure}

\subsection{Protein Signaling Network}\label{sec:sachsExample}

\citet[Figure 2]{sachs2005causal} present a signaling network of 11
observed molecules and 13 unobserved molecules. The black edges in
Figure \ref{fig:protein_network} give a plausible mixed graph
representation of that network and was also considered by
\citet{drton2016computation}.  A log-transformation of the available
protein expression data improves Gaussianity but leaves the
distribution of some of the variables skewed and/or multimodal.  We consider two separate tests; each compares the SEM sub-model corresponding to the graph of black edges against a full model which adds one of the two red edges also shown in Figure
\ref{fig:protein_network}. Note that the added red edge from $\text{Mek} \rightarrow \text{PKA}$ induces a directed cycle.   For
the log-transformed data, we perform a Gaussian as well as an
empirical likelihood ratio test.  For the test involving the directed edge $\text{Mek} \rightarrow \text{PKA}$, the Gaussian LR is .416 (p-value = .52) and the ELR is 4.379 (p-value = .04). For the test involving the bidirected edge $\text{Akt} \leftrightarrow \text{PIP2}$, the Gaussian LR is 15.216 (p-value < .001) and the ELR is .782 (p-value = .37).  While we do not have a
certified gold standard network, and the implicit assumption of
linearity may not be appropriate for all postulated relationships,
these examples present situations in which the Gaussian assumption is
particularly inappropriate and may cause concerns for a practitioner.

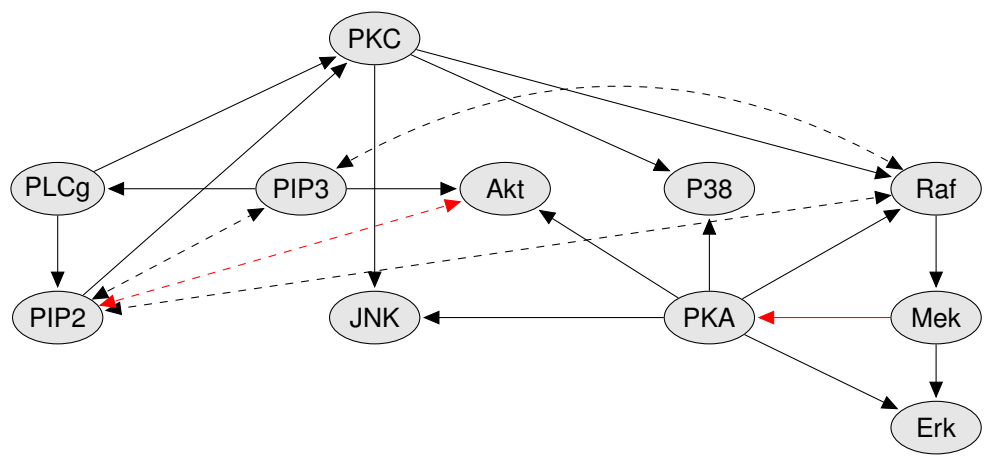
\begin{figure}[t]
  \centering
	\begin{tikzpicture}[->,>=triangle 45,shorten >=1pt,
            auto,
            main node/.style={ellipse,inner sep=0pt,fill=gray!20,draw,font=\sffamily,
             minimum width = 1.2cm, minimum height = .7cm}]

            \node[main node] (PKC) {PKC};
            \node[main node] (JNK) [below=3cm of PKC] {JNK};
			\node[main node] (PIP2) [left=3cm of JNK] {PIP2};
            \node[main node] (PLCg) [above=1cm of PIP2] {PLCg};
          	\node[main node] (PIP3) [right=2cm of PLCg] {PIP3};
          	\node[main node] (Akt) [right=1.5cm of PIP3] {Akt};
          	\node[main node] (P38) [right=1.5cm of Akt] {P38};
          	\node[main node] (Raf) [right=1.8cm of P38] {Raf};
          	\node[main node] (PKA) at (P38 |- JNK) {PKA};
          	\node[main node] (Mek) at (PKA -| Raf) {Mek};
          	\node[main node] (Erk) [below=.75cm of Mek] {Erk};

          	 \path[color=black,style={->}]
          	 (PKC) edge node {} (JNK)
          	 (PKC) edge node {} (Raf)
          	 (PKC) edge node {} (P38)

          	 (PLCg) edge node {} (PIP2)
			       (PIP3) edge node {} (PLCg)
			       (PIP3) edge node {} (Akt)
		   		   (Raf) edge node {} (Mek)
       		(PKA) edge node {} (JNK)
       		(PKA) edge node {} (P38)
       		(PKA) edge node {} (Erk)
       		(Mek) edge node {} (Erk)
       	  (PIP2) edge node {} (PKC)
          (PKA) edge node {} (Akt)
          (PLCg) edge node {} (PKC)
          (PKA) edge node {} (Raf);

       		\path[color=black,style={<->, dashed}]
			       (PIP2) edge node {} (PIP3)
			          (Raf) edge node {} (PIP2)
       		;

          \path[color=red,style={->}]
             (Mek) edge node {} (PKA);

       	\path[color=black,style={<->, dashed, bend right}]
       					(Raf) edge node {} (PIP3)
       		       		;

        \path[color=red,style={<->, dashed}]
                (PIP2) edge node {} (Akt)
                    ;

            \end{tikzpicture}
  \caption{Plausible mixed graph for the protein-signaling network dataset. The relevant sub-model can be formed by removing the red bidirected edge between PIP2 and Akt and the red edge from MEK to PKA.}
  \label{fig:protein_network}
\end{figure}

\section{Discussion}\label{sec:discussion}

In this article, we showed that EL methods are an attractive
alternative for estimation and testing of non-Gaussian linear SEMs.
Our approach of profiling out the error covariance matrix $\Omega$
drastically reduces computational effort and creates a far more
tractable and reliable estimation procedure.  Furthermore, we showed
that the use of AEL may further improve convergence of optimizers,
particularly, when the sample size is small and the errors are skewed.
EEL was seen to drastically improve the
coverage rate of the joint confidence intervals.

Our EL methods are applicable under very few distributional
assumptions, all the while allowing statistical inference in close to
analogy to parametric modeling.  When the data is non-Gaussian, the
modified EL methods outperform the other methods we considered in
almost all scenarios we explored.  This concerns the proportion of
times a valid estimate is returned, statistical efficiency, and also
confidence region coverage.
While there remains significant room for improvement in the design of
confidence regions,
we conclude that EL methods are a valuable tool for applications of
linear SEMs to non-Gaussian data.

\bigskip
{\bf Appendix}
\smallskip

\begin{proof}[Proof of Proposition~\ref{lem:conditions}]
  We recall our notation $\theta=(B,\Omega)$ and
  $\theta_0=(B_0,\Omega_0)$.  Based on the right-most expression
  in~(\ref{eq:equiv-equations}), the considered naive/direct
  estimating equations may be based on the function $G(y,B)$ with
  coordinates
  \begin{align*}
    G_v(y, B) &= y_v, & G_{uv}(y,B) &= \left(y_v-
      \sum_{s\in\pa(v)}\beta_{vs}y_s\right)\left(y_u-
      \sum_{t\in\pa(u)}\beta_{ut}y_t\right) -\omega_{uv},
  \end{align*}
  for $v\in V$ and $\{u,v\} \in V \times V$, respectively.

  Our claim follows from Theorem 1 of \citet{qin1994empirical} under
  the following conditions:
  \begin{enumerate}[label=(\arabic*)]
  \item $\mathbb{E}(G(Y_i, \theta_0)G(Y_i, \theta_0)^t)$ is positive definite.
  \item In a neighborhood of the d-dimensional parameter $\theta_0$,
    the derivative $\frac{\partial G(y, \theta)}{\partial \theta}$ is
    continuous, and
    $\left\Vert \frac{\partial G(y, \theta)}{\partial \theta}
    \right\Vert$ and $\left \Vert G(y, \theta)\right\Vert^3$ are
    bounded by an integrable function $M_1(y)$.
  \item $\mathbb{E}\left( \frac{\partial G(Y_i, \theta_0)}{\partial \theta}\right)$ has rank $d$.
  \item $\partial^2 G(y, \theta) / \partial \theta \theta^T$ is continuous and $\left\Vert\partial^2 G(y, \theta) / \partial \theta \theta^T \right\Vert$ is bounded by an integrable function $M_2(y)$ in a neighborhood of the true parameter $\theta_0$.
  \end{enumerate}
  Here, $\left\Vert \cdot \right\Vert$ denotes the Euclidean norm.

  Noting that $G_{uv}(Y_i,B_0)=\epsilon_v\epsilon_u-\omega_{uv}$,
  condition (1) is an immediate consequence of assumption (b) in our
  proposition.  Condition (3) is implied by assumption (a).  With
  polynomial estimating equations, all derivatives in conditions (2)
  and (4) exist.  Now, $G$ and its first and second partial derivatives are at most quadratic
  functions of $Y_i$, which in turn is a linear function of a
  realization of the error vector $\epsilon$.  Local bounds on the
  concerned quantities are easily obtained and assumption (b) ensures
  their integrability.
\end{proof}

\section{Acknowledgments}

This material is based upon work supported by the U.S.~National Science
Foundation under Grants No.~DMS 1561814 and 1712535. We would also like to thank the Institute for Mathematical Sciences at the National University of Singapore for travel support to the 2016 workshop: Empirical Likelihood Based Methods in Statistics.

\bibliographystyle{abbrvnat}
\bibliography{empLikSEM}

\end{document}